\documentclass[a4paper]{article}
\usepackage{lmodern}
\pdfoutput=1
\usepackage{graphicx,wrapfig}
\usepackage{lipsum}
\usepackage{amsmath,amssymb,mathrsfs}
\usepackage{hyperref}
\usepackage{url}
\usepackage{mathtools}
\usepackage{changepage}
\usepackage[ruled,vlined]{algorithm2e}
\usepackage{amsmath}
\usepackage{multirow}
\usepackage{wrapfig}
\usepackage{subfig}
\usepackage{color}
\usepackage{epsfig,enumerate,amsmath,amsfonts,amssymb,amsthm,mathrsfs,ifpdf}
\usepackage{indentfirst,relsize}
\usepackage{setspace,graphicx}
\usepackage{latexsym}
\usepackage[all]{xy}
\usepackage[dvipsnames]{pstricks}
\usepackage{pst-grad} 
\usepackage{pst-plot} 
\usepackage[margin = 2.50cm]{geometry}

\usepackage{authblk} 
\usepackage{algpseudocode}
\usepackage{color,soul}
\usepackage{enumitem}

\usepackage{orcidlink}
\usepackage{tabularx}
\usepackage{pst-grad} 
\usepackage{pst-plot} 
\usepackage{enumitem}

\usepackage{orcidlink}

\usepackage{tikz}
\usetikzlibrary{arrows.meta,positioning}

\newcommand{\remove}[1]{}

\newtheorem{theorem}{Theorem}
\newtheorem{lemma}[theorem]{Lemma}

\newtheorem{corollary}[theorem]{Corollary}

\newtheorem{observation}[theorem]{Observation}

\newcommand{\DP}{\textsc{DP}}

\usepackage{todonotes}
\usepackage[most]{tcolorbox}

\title{Maximum Edge Open Packing in Permutation, Interval, and
	Well-Partitioned Chordal Graphs}

\author[1]{Gautam K. Das\footnote{gkd@iitg.ac.in}}
\author[1]{Kamal Santra \orcidlink{0009-0006-5997-1452} \footnote{kamal.7.2013@gmail.com, kamal.santra@iitg.ac.in}}
\affil[1]{Department of Mathematics\\
	
	Indian Institute of Technology Guwahati\\
	
	Guwahati, 781039, Assam, India}

\date{}

\begin{document}

	\maketitle
\begin{abstract}
Edge open packing is a relaxation of induced matching in which the selected
edges may induce disjoint stars. We study the \textsc{Maximum Edge Open
	Packing} problem on permutation graphs, interval graphs, and well-partitioned
chordal graphs. For the first two classes, we introduce an oriented
star-conflict graph whose vertices are ordered edges. We prove that its
compatibility graph admits a natural transitive orientation: a product-order
orientation for permutation graphs and a left-to-right orientation for
interval graphs. In each case, a maximum edge open packing is obtained from a
maximum clique, equivalently a longest directed path, in the compatibility
graph. Given the corresponding representation, both algorithms run in
\(O(n^2+m^2)\leq O(n^4)\) time, where \(n=|V(G)|\) and \(m=|E(G)|\). For
well-partitioned chordal graphs, we give a
dynamic program over a partition tree. Its states use the fact that the
endpoint set of an edge open packing meets each clique bag in at most two
vertices. Given a partition-tree representation, the edge open packing number
is computed in \(O(n^4)\) time, and an optimal
packing can be reconstructed within the same time bound.
\end{abstract}
	


\section{Introduction}
\label{sec:introduction}

Packing problems are central objects in graph algorithms and combinatorial
optimization. In many applications, the aim is to select a collection of
vertices or edges whose members do not interfere with one another. Classical
examples include matchings, induced matchings, open packings, and several
packing variants of graph coloring. Among these, the induced matching problem
asks for a set of edges whose endpoints induce a matching. This strong
separation condition makes the problem algorithmically difficult: finding a
maximum induced matching is NP-hard in general and remains hard on several
restricted graph classes
\cite{cameron1989induced,dabrowski2013new,lozin2002maximum,stockmeyer1982np}.

Edge open packing, introduced by Chelladurai et
al.~\cite{chelladurai2022edge}, is a natural relaxation of induced matching.
An edge set is an edge open packing if no two of its edges have a third edge
joining an endpoint of one to an endpoint of the other. Equivalently, the
subgraph induced by the endpoints of the selected edges is a disjoint union of
induced stars. Thus, every induced matching is an edge open packing, whereas an
edge open packing may contain several edges that share the centre of the same
induced star.

The algorithmic study of edge open packing was initiated by Bre{\v{s}}ar and
Samadi~\cite{brevsar2024edge}. They proved that the decision problem is
NP-complete even for graphs with a universal vertex, Eulerian bipartite graphs,
and planar graphs of maximum degree four, and they gave a linear-time algorithm
for trees. Bre{\v{s}}ar et al.~\cite{BRESAR2025} subsequently investigated the
relationship between induced matchings and edge open packings, particularly on
trees and graph products. Pandey and Santra~\cite{pandey2025edge} obtained
further structural characterizations, while Santra~\cite{eopblock} developed
an \(O(n^2)\)-time algorithm for proper interval graphs, a linear-time
algorithm for block graphs, and an \(O(n^3)\)-time algorithm for split graphs.
More recently, Bhyravarapu et al.~\cite{BhyravarapuDasSantra2026} gave polynomial-time algorithms for
distance-hereditary and biconvex bipartite graphs. They also gave an
\(2^{O(\omega(G))}\operatorname{poly}(n)\)-time algorithm for chordal graphs,
thereby showing that the problem is fixed-parameter tractable when parameterized
by the clique number.

Bre{\v{s}}ar and Samadi~\cite{brevsar2024edge} asked whether the edge open
packing number can be computed in polynomial time on chordal graphs. The
fixed-parameter algorithm mentioned above does not settle this question when
the clique number is unbounded. Interval graphs form a natural superclass of
proper interval graphs, and well-partitioned chordal graphs form a natural
superclass of split graphs. Since both are subclasses of chordal graphs,
polynomial-time algorithms for these classes provide further partial progress
toward resolving the question.

In this paper, we study the \textsc{Maximum Edge Open Packing} problem on
permutation graphs, interval graphs, and well-partitioned chordal graphs. For
the first two classes, we use an auxiliary graph built from ordered edges. For
a graph \(G\), the oriented star-conflict graph \(A_G\) has one vertex for each
ordered edge \((c,\ell)\), interpreted as selecting the edge \(c\ell\) with
\(c\) prescribed as a star centre and \(\ell\) prescribed as a leaf. Two
vertices of \(A_G\) are adjacent precisely when their ordered edges cannot
occur together with these prescribed roles. Forgetting the orientations in an
independent set of \(A_G\) produces an edge open packing of the same size.
Conversely, after choosing a centre in every induced-star component, every edge
open packing can be oriented to give an independent set of \(A_G\). It follows
that \(\rho_e^o(G)=\alpha(A_G)\).

Our first contribution concerns permutation graphs. We prove that, when \(G\)
is a permutation graph, the compatibility graph
\(H_G=\overline{A_G}\) admits a natural transitive orientation. Hence,
\(H_G\) is a comparability graph, and a maximum clique of \(H_G\) can be found
as a longest directed path in this orientation. Given a permutation
representation of \(G\), this yields an
\(O(n^2+m^2)\le O(n^4)\)-time algorithm.

Our second contribution applies the same ordered-edge framework to interval
graphs. We show that, for an interval graph \(G\), the compatibility graph
\(H_G\) admits a transitive left-to-right orientation derived from an interval
representation. Consequently, \(\rho_e^o(G)\) is again obtained from a maximum
clique in a comparability graph. Given an interval representation of \(G\), the
resulting algorithm runs in \(O(n^2+m^2)\le O(n^4)\) time. Since interval
graphs properly generalize proper interval graphs, this extends the previously
known polynomial-time result to a larger graph class.

Our third contribution concerns well-partitioned chordal graphs, a subclass of
chordal graphs introduced by Ahn et al.~\cite{AhnJaffkeKwonLima2022}. Such a
graph admits a partition of its vertex set into clique bags whose interactions
are described by a tree. We exploit this partition tree in a dynamic program.
The key structural fact is that, for every edge open packing \(D\) and every
clique bag \(X\), at most two vertices of \(X\) belong to \(V(D)\). This gives
only quadratically many boundary states per bag. When a partition-tree
representation is supplied, the algorithm computes the optimum value in
\(O(n^4)\) time. Since every split graph is
well-partitioned chordal, this extends the known tractability result for split
graphs to a broader chordal subclass.

The remainder of the paper is organized as follows.
Section~\ref{sec:preliminaries} introduces the notation and basic definitions.
Section~\ref{sec:permutation} treats permutation graphs using the oriented
star-conflict graph, and Section~\ref{sec:interval} develops the corresponding
algorithm for interval graphs. Section~\ref{sec:wpcg} presents the dynamic
program for well-partitioned chordal graphs. Finally,
Section~\ref{sec:conclusion} summarizes the results and discusses directions
for further research.

\section{Preliminaries}
\label{sec:preliminaries}

All graphs considered in this paper are finite, simple, and undirected. For a
graph \(G\), we write \(V(G)\) and \(E(G)\) for its vertex set and edge set,
respectively, and put \(n=|V(G)|\) and \(m=|E(G)|\). The independence number
and clique number of \(G\) are denoted by \(\alpha(G)\) and \(\omega(G)\).
For a vertex \(v\in V(G)\), the open and closed neighbourhoods of \(v\) are
denoted by \(N_G(v)\) and \(N_G[v]\), respectively. For a set
\(X\subseteq V(G)\), let \(G[X]\) denote the subgraph of \(G\) induced by
\(X\). If \(X,Y\subseteq V(G)\) are disjoint, then
\(E_G(X,Y)=\{xy\in E(G):x\in X,\ y\in Y\}\) is the set of edges with one
endpoint in \(X\) and the other in \(Y\). A \emph{star} is a graph
isomorphic to \(K_{1,t}\) for some \(t\ge 1\). For \(t\geq 2\), its unique
vertex of degree \(t\) is called the centre, and its other vertices are called
leaves. When \(t=1\), either endpoint may be chosen as the centre.

Let \(D\subseteq E(G)\). We write \(V(D)\) for the set of vertices incident
with the edges of \(D\). Thus, \(G[V(D)]\) is the subgraph induced by all
endpoints of the edges in \(D\). Two distinct edges \(e_1,e_2\in E(G)\) have a
\emph{common edge} if there is an edge
\(e\in E(G)\setminus\{e_1,e_2\}\) joining an endpoint of \(e_1\) to an
endpoint of \(e_2\). An edge set \(D\subseteq E(G)\) is an \emph{edge open
packing} if no two edges of \(D\) have a common edge. The maximum cardinality
of an edge open packing of \(G\) is the \emph{edge open packing number} of
\(G\), denoted by \(\rho_e^o(G)\).

We use the following characterization throughout the paper.

\begin{observation}\label{obs:eop-stars}
	An edge set \(D\subseteq E(G)\) is an edge open packing if and only if
	\(G[V(D)]\) is a disjoint union of induced stars.
\end{observation}

\begin{proof}
	Suppose first that \(D\) is an edge open packing, and put
	\(H=G[V(D)]\). We claim that \(E(H)=D\). The inclusion
	\(D\subseteq E(H)\) is immediate. If some edge
	\(xy\in E(H)\setminus D\) existed, then there would be
	edges \(e_x,e_y\in D\) incident with \(x\) and \(y\), respectively. These
	two selected edges are distinct: if \(e_x=e_y\), then this edge has endpoints
	\(x\) and \(y\), and hence it is the edge \(xy\), contrary to
	\(xy\notin D\). The edge \(xy\) would therefore be a common edge of
	\(e_x\) and \(e_y\), a contradiction. This proves the claim.

	Let \(C\) be a connected component of \(G[V(D)]\). No edge of \(C\) can
	have both endpoints of degree at least two in \(C\). Indeed, suppose that
	\(uv\in E(C)\) and that both \(u\) and \(v\) have degree at least two.
	Choose edges \(ux\) and \(vy\) of \(C\), where \(x\ne v\) and
	\(y\ne u\). By the preceding claim, all three edges belong to \(D\). If
	\(x\ne y\), then \(uv\) is a common edge of the distinct selected edges
	\(ux\) and \(vy\). If \(x=y\), then \(uv\) is a common edge of the
	distinct selected edges \(ux\) and \(vx\). Both cases contradict the choice
	of \(D\).

	If \(C\) consists of one edge, then it is \(K_{1,1}\). Otherwise, \(C\)
	has a vertex \(c\) of degree at least two. Every neighbour of \(c\) has
	degree one by the preceding paragraph. Since \(C\) is connected, it follows
	that every vertex different from \(c\) is a neighbour of \(c\). Hence,
	\(C\) is a star. As \(G[V(D)]\) is an induced subgraph of \(G\), its
	components are induced stars. Thus, \(G[V(D)]\) is a disjoint union of
	induced stars.

	Conversely, suppose that \(G[V(D)]\) is a disjoint union of induced stars,
	and let \(e_1,e_2\in D\) be distinct. If they lie in different components,
	then no edge of \(G\) joins an endpoint of \(e_1\) to an endpoint of
	\(e_2\), since such an edge would also belong to \(G[V(D)]\). If they lie
	in the same star component, then they share its centre, and the only edges
	between their endpoint sets are \(e_1\) and \(e_2\) themselves. Therefore,
	\(e_1\) and \(e_2\) have no common edge. Since the pair was arbitrary,
	\(D\) is an edge open packing.
\end{proof}

A graph is a \emph{permutation graph} if its vertices admit two linear orders
\(<_1\) and \(<_2\) such that two vertices are adjacent precisely when they
occur in opposite orders in \(<_1\) and \(<_2\). Equivalently, permutation
graphs are the incomparability graphs of two-dimensional partial orders. A
graph is an \emph{interval graph} if its vertices can be represented by closed
intervals on the real line such that two vertices are adjacent precisely when
their intervals intersect.

An orientation of a graph is \emph{transitive} if \(x\to y\) and \(y\to z\)
imply that \(xz\) is an edge oriented as \(x\to z\). A graph is a
\emph{comparability graph} if it admits a transitive orientation, and it is a
\emph{cocomparability graph} if its complement is a comparability graph. We
use standard facts about these graph classes; see
Golumbic~\cite{Golumbic1980}.

A graph is \emph{chordal} if every cycle of length at least four has a chord.
Chordal graphs admit clique trees: the maximal cliques can be arranged as the
nodes of a tree such that, for every vertex \(v\), the maximal cliques
containing \(v\) induce a connected subtree
\cite{FulkersonGross1965,Golumbic1980}.

A connected graph is a \emph{well-partitioned chordal graph} if its vertex set
admits a partition \(\mathcal P\) into cliques together with a tree \(T\) whose
nodes are the bags in \(\mathcal P\), subject to the following conditions. For
every two distinct bags \(X,Y\), if \(XY\in E(T)\), there are sets
\(\operatorname{bd}(X,Y)\subseteq X\) and
\(\operatorname{bd}(Y,X)\subseteq Y\) such that
\(E_G(X,Y)=\operatorname{bd}(X,Y)\times
\operatorname{bd}(Y,X)\); if \(XY\notin E(T)\), then
\(E_G(X,Y)=\emptyset\). The pair
\((\mathcal P,T)\) is called a \emph{partition-tree representation} of
\(G\). This graph class was introduced by Ahn et
al.~\cite{AhnJaffkeKwonLima2022} and contains all split graphs. We adopt the
standard componentwise convention for disconnected graphs: a disconnected
graph is well-partitioned chordal if each of its connected components admits a
partition-tree representation. The disjoint union of the corresponding trees
is then a partition forest for the graph.

\section{Permutation Graphs}
\label{sec:permutation}

In this section, we show that the
\textnormal{\textsc{Maximum Edge Open Packing}} problem
can be solved in polynomial time on permutation graphs. The proof uses
the oriented star-conflict graph. The key point is that, when \(G\) is a
permutation graph, the complement \(H_G=\overline{A_G}\) admits a natural
transitive orientation. Hence, \(H_G\) is a comparability graph, and the desired
solution can be found as a maximum clique of \(H_G\).

\subsection{The oriented star-conflict graph}

Let \(G\) be a graph. We define the \emph{oriented star-conflict graph} of
\(G\), denoted by \(A_G\), as follows. The vertices of \(A_G\) are the ordered
edges of \(G\); that is, \(V(A_G)=\{(c,\ell):c\ell\in E(G)\}\). We interpret
an ordered edge \((c,\ell)\) as selecting the edge \(c\ell\), where \(c\) is
viewed as the centre of a star and \(\ell\) is viewed as a leaf.

Two ordered edges \((c,\ell)\) and \((d,r)\) are called compatible if they can
appear together in one edge open packing with the prescribed orientations. This
happens in exactly two cases. First, they may have the same centre: \(c=d\),
\(\ell\neq r\), and \(\ell r\notin E(G)\). In this case, the two selected
edges belong to the same induced star centred at \(c\). Second, they may have
different centres. In that case, the four vertices \(c,\ell,d,r\) must be
distinct, and there must be no edge between the endpoint sets \(\{c,\ell\}\)
and \(\{d,r\}\). Equivalently, we require \(cd,cr,\ell d,\ell r\notin E(G)\).
In this case, the two selected edges belong to two different star components.

In particular, if \(c\neq d\), then the ordered edges \((c,\ell)\) and
\((d,\ell)\) are not compatible, because they share the same leaf. They do not
satisfy the same-centre condition, and they also fail the different-centre
condition since the four endpoints are not distinct.

The graph \(A_G\) records conflicts between ordered edges. Thus, two vertices
of \(A_G\) are adjacent if and only if the corresponding ordered edges are not
compatible.

We now relate edge open packings of \(G\) to independent sets of \(A_G\).
The vertices of \(A_G\) represent edges of \(G\) together with prescribed
centre--leaf roles, and two vertices of \(A_G\) are adjacent exactly when the
corresponding ordered edges are incompatible with these roles. Thus, choosing
pairwise compatible ordered edges is the same as choosing pairwise nonadjacent
vertices of \(A_G\).

\begin{lemma}\label{lem:eop-independent-ag}
	For every graph \(G\), \(\rho_e^o(G)=\alpha(A_G)\).
\end{lemma}

\begin{proof}
	Let \(D\) be an edge open packing of \(G\). Then \(G[V(D)]\) is a
	disjoint union of induced stars. For every edge of \(D\), choose an orientation
	as follows. If its star component has at least two edges, orient it from the
	unique centre of that star to its leaf. If the component is a single edge
	\(K_2\), choose either endpoint as its centre and orient the edge accordingly.
	Let \(I_D\) be the set of all ordered edges obtained in this way. Thus, \(I_D\)
	contains exactly one ordered edge for each edge of \(D\), and hence,
	\(|I_D|=|D|\).

	Any two ordered edges in \(I_D\) are compatible. If they belong to the same
	star, then they have the same centre and their leaves are nonadjacent because
	the star is induced. If they belong to different star components, then there
	is no edge between the two components. Hence, \(I_D\) is an independent set of
	\(A_G\), and therefore, \(\rho_e^o(G)\le \alpha(A_G)\).

	Conversely, let \(I\) be an independent set of \(A_G\), and define
	\(D_I=\{c\ell:(c,\ell)\in I\}\). Since \(I\) is independent, any two ordered
	edges in \(I\) are compatible. All ordered edges with the same first coordinate
	form an induced star with that first coordinate as the centre. Ordered edges
	with different first coordinates are vertex-disjoint and have no edge between
	their endpoint sets. Moreover, the map from \(I\) to \(D_I\) is injective.
	Indeed, the only two distinct ordered edges representing the same undirected
	edge \(c\ell\) are \((c,\ell)\) and \((\ell,c)\), and these two ordered edges
	are incompatible because their centres are different while their endpoint sets
	are not disjoint. Hence, an independent set cannot contain both of them.
	Therefore, \(G[V(D_I)]\) is a disjoint union of induced stars, \(D_I\) is an
	edge open packing of \(G\), and \(|D_I|=|I|\).
	Thus \(\alpha(A_G)\le \rho_e^o(G)\), and consequently
	\(\rho_e^o(G)=\alpha(A_G)\).
\end{proof}

\subsection{Permutation graphs and the product order}

Let \(G\) be a permutation graph. Fix a permutation representation of \(G\),
or equivalently two linear orders \(<_{1}\) and \(<_{2}\) on \(V(G)\), such
that two vertices are adjacent in \(G\) if and only if they appear in different
orders in the two linear orders.

Define a partial order \(\prec\) on \(V(G)\) by \(u\prec v\) if and only if
\(u<_{1}v\) and \(u<_{2}v\).
Then two vertices are adjacent in \(G\) exactly when they are incomparable
under \(\prec\). Equivalently, non-edges of \(G\) are precisely comparable pairs
of the poset \((V(G),\prec)\).

We shall use the following elementary property of two-dimensional partial
orders.

\begin{lemma}\label{lem:two-crossing-blocks}
	Let \((V,\prec)\) be a partial order obtained from two linear orders. Let
	\(\{a,b\}\) and \(\{c,d\}\) be two incomparable pairs. Suppose that every
	vertex of \(\{a,b\}\) is comparable with every vertex of \(\{c,d\}\). Then one
	of the two pairs is completely before the other. More precisely, either
	\(a\prec c\), \(a\prec d\), \(b\prec c\), and \(b\prec d\), or
	\(c\prec a\), \(c\prec b\), \(d\prec a\), and \(d\prec b\).
\end{lemma}

\begin{proof}
	Let \(z\) be a vertex comparable with both \(a\) and \(b\). We claim that \(z\)
	must precede both \(a\) and \(b\), or follow both of them. There are four
	possible combinations of the comparisons between \(z\) and \(a,b\). The mixed
	combination \(a\prec z\prec b\) would imply \(a\prec b\), and the other mixed
	combination \(b\prec z\prec a\) would imply \(b\prec a\), by transitivity.
	Both conclusions contradict the incomparability of \(a\) and \(b\). Thus,
	either \(z\prec a\) and \(z\prec b\), or \(a\prec z\) and \(b\prec z\).

	Apply this claim to the vertices \(c\) and \(d\). Both \(c\) and \(d\) are
	comparable with both \(a\) and \(b\), so each of them is positioned either
	completely before the pair \(\{a,b\}\) or completely after it. If one of
	\(c,d\) preceded \(\{a,b\}\) and the other followed \(\{a,b\}\), then
	transitivity through either \(a\) or \(b\) would make \(c\) and \(d\)
	comparable, contradicting the assumption that \(\{c,d\}\) is an incomparable
	pair. Hence, \(c\) and \(d\) are positioned on
	the same side of \(\{a,b\}\). Thus either both \(c,d\) are before both
	\(a,b\), or both \(a,b\) are before both \(c,d\). This proves the lemma.
\end{proof}

\subsection{The complement of \texorpdfstring{\(A_G\)}{A\_G}}

Let \(H_G=\overline{A_G}\). Thus, \(V(H_G)=V(A_G)\), and two vertices of
\(H_G\) are adjacent exactly when the corresponding ordered edges of \(G\) are
compatible. We prove that, for permutation graphs, \(H_G\) is a comparability
graph. Recall that a graph is a \emph{cocomparability graph} if its complement
is a comparability graph. Consequently, \(A_G\) is a cocomparability graph.

We now define an orientation \(\triangleleft\) of the edges of \(H_G\). Thus,
whenever we consider two vertices \(e=(c,\ell)\) and \(f=(d,r)\) below, we
assume that they are adjacent in \(H_G\). Equivalently, the ordered edges
\((c,\ell)\) and \((d,r)\) are compatible. In particular, if \(c\neq d\), then
the four endpoints \(c,\ell,d,r\) are distinct. Hence, pairs such as
\((c,\ell)\) and \((d,\ell)\), with \(c\neq d\), are not edges of \(H_G\), and
no orientation has to be assigned to them.

Let \(e=(c,\ell)\) and \(f=(d,r)\) be two adjacent vertices of \(H_G\).
First suppose that \(c=d\). Then compatibility implies that
\(\ell r\notin E(G)\). Since nonedges of a permutation graph are exactly
comparable pairs under \(\prec\), exactly one of \(\ell\prec r\) and
\(r\prec \ell\) holds. If \(\ell\prec r\), we orient the edge from
\((c,\ell)\) to \((c,r)\), and write \((c,\ell)\triangleleft(c,r)\).
Otherwise, we have \(r\prec \ell\), and we orient the edge from \((c,r)\) to
\((c,\ell)\), writing \((c,r)\triangleleft(c,\ell)\).

Now suppose \(c\neq d\). Since \(e\) and \(f\) are adjacent in \(H_G\), the
ordered edges \((c,\ell)\) and \((d,r)\) are compatible. Hence, there is no edge
between the endpoint sets \(\{c,\ell\}\) and \(\{d,r\}\). Therefore, every
vertex of \(\{c,\ell\}\) is comparable with every vertex of \(\{d,r\}\) under
\(\prec\). Also, since \(c\ell\in E(G)\) and \(dr\in E(G)\), the pairs
\(\{c,\ell\}\) and \(\{d,r\}\) are incomparable pairs under \(\prec\). By
Lemma~\ref{lem:two-crossing-blocks}, one of these two pairs is positioned
completely before the other.

If \(c\prec d\), \(c\prec r\), \(\ell\prec d\), and \(\ell\prec r\), then we
orient the edge from \((c,\ell)\) to \((d,r)\), that is,
\((c,\ell)\triangleleft(d,r)\). Otherwise, the second endpoint pair is
positioned completely before the first one; equivalently, \(d\prec c\),
\(d\prec \ell\), \(r\prec c\), and \(r\prec \ell\). In this case, we orient
the edge from \((d,r)\) to \((c,\ell)\), that is,
\((d,r)\triangleleft(c,\ell)\).

This defines exactly one direction for every edge of \(H_G\).

\begin{lemma}\label{lem:orientation-transitive}
	The orientation \(\triangleleft\) of \(H_G\) is transitive.
\end{lemma}

\begin{proof}
	Let \(e=(c,\ell)\), \(f=(d,r)\), and \(g=(a,b)\) be vertices of \(H_G\) such
	that \(e\triangleleft f\) and \(f\triangleleft g\). We show that
	\(e\triangleleft g\). By the definition of the orientation, each comparison is
	of one of two types: either the two ordered edges have the same centre, or
	their endpoint pairs are distinct, and one endpoint pair is positioned
	completely before the other in the product order \(\prec\).

	First, suppose that both \(e\triangleleft f\) and \(f\triangleleft g\) are of
	different-centre type. Then the endpoint pair \(\{c,\ell\}\) is positioned
	completely before the endpoint pair \(\{d,r\}\), and the endpoint pair
	\(\{d,r\}\) is positioned completely before the endpoint pair \(\{a,b\}\).
	By the transitivity of \(\prec\), the endpoint pair \(\{c,\ell\}\) is
	positioned completely before the endpoint pair \(\{a,b\}\). Hence, every vertex
	of \(\{c,\ell\}\) is comparable with every vertex of \(\{a,b\}\). Therefore,
	there is no edge of \(G\) between the two endpoint sets. Also, since
	\(c\ell\in E(G)\) and \(ab\in E(G)\), both endpoint pairs are incomparable
	under \(\prec\). Thus \((c,\ell)\) and \((a,b)\) are compatible, so \(e\) and
	\(g\) are adjacent in \(H_G\). Since \(\{c,\ell\}\) is positioned completely
	before \(\{a,b\}\), the orientation rule gives \(e\triangleleft g\).

	Next, suppose that \(e\triangleleft f\) is of same-centre type and
	\(f\triangleleft g\) is of different-centre type. Then \(e=(c,\ell)\) and
	\(f=(c,r)\), with \(\ell\prec r\). Also, since \(f\triangleleft g\) is of
	different-centre type, the endpoint pair \(\{c,r\}\) is positioned completely
	before the endpoint pair \(\{a,b\}\). Hence, \(c\prec a\), \(c\prec b\),
	\(r\prec a\), and \(r\prec b\). Since \(\ell\prec r\), we also have
	\(\ell\prec a\) and \(\ell\prec b\). Therefore,  the endpoint pair
	\(\{c,\ell\}\) is positioned completely before \(\{a,b\}\). Hence
	\((c,\ell)\) and \((a,b)\) are compatible, \(e\) and \(g\) are adjacent in
	\(H_G\), and the orientation rule gives \(e\triangleleft g\).

	Now suppose that \(e\triangleleft f\) is of the different-centre type and
	\(f\triangleleft g\) is of the same-centre type. Then the endpoint pair
	\(\{c,\ell\}\) is positioned completely before the endpoint pair \(\{d,r\}\).
	Also, \(f=(d,r)\) and \(g=(d,b)\), with \(r\prec b\). Thus \(c\prec d\),
	\(c\prec r\), \(\ell\prec d\), and \(\ell\prec r\). Since \(r\prec b\), we get
	\(c\prec b\) and \(\ell\prec b\). Hence, the endpoint pair \(\{c,\ell\}\) is
	positioned completely before \(\{d,b\}\). Therefore \((c,\ell)\) and
	\((d,b)\) are compatible, \(e\) and \(g\) are adjacent in \(H_G\), and the
	orientation rule gives \(e\triangleleft g\).

	Finally, suppose that both \(e\triangleleft f\) and \(f\triangleleft g\) are of
	the same-centre type. Then all three ordered edges have the same centre. Thus
	\(e=(c,\ell)\), \(f=(c,r)\), and \(g=(c,b)\), with \(\ell\prec r\) and
	\(r\prec b\). By the transitivity of \(\prec\), we have \(\ell\prec b\). Hence
	\(\ell b\notin E(G)\), so \((c,\ell)\) and \((c,b)\) are compatible. Thus,
	\(e\) and \(g\) are adjacent in \(H_G\), and the same-centre orientation rule
	gives \(e\triangleleft g\).

	In every case, \(e\triangleleft f\) and \(f\triangleleft g\) imply
	\(e\triangleleft g\). Therefore, the orientation \(\triangleleft\) is
	transitive.
\end{proof}

\begin{theorem}\label{thm:ag-permutation-cocomparability}
	If \(G\) is a permutation graph, then \(A_G\) is a cocomparability graph.
\end{theorem}

\begin{proof}
	The preceding construction gives an orientation of every edge of
	\(H_G=\overline{A_G}\). By Lemma~\ref{lem:orientation-transitive}, this
	orientation is transitive. Hence, \(H_G\) is a comparability graph. Therefore,
	\(A_G=\overline{H_G}\) is a cocomparability graph.
\end{proof}

\subsection{Algorithmic consequence}

We now obtain the algorithm for permutation graphs.

\begin{theorem}\label{thm:eop-permutation-poly}
	Given a permutation representation of a permutation graph \(G\), the
	\textnormal{\textsc{Maximum Edge Open Packing}} problem can be solved in
	\(O(n^2+m^2)\) time.
\end{theorem}

\begin{proof}
	Let \(G\) be a permutation graph with a given permutation representation. If
	\(E(G)\) is empty, the empty set is a maximum edge open packing, so assume that
	\(E(G)\neq\emptyset\). Construct the oriented star-conflict graph
	\(A_G\). By Lemma~\ref{lem:eop-independent-ag}, we have
	\(\rho_e^o(G)=\alpha(A_G)\). By
	Theorem~\ref{thm:ag-permutation-cocomparability}, \(A_G\) is a cocomparability
	graph, or equivalently \(H_G=\overline{A_G}\) is a comparability graph. Hence,
	\(\alpha(A_G)=\omega(H_G)\). Since \(H_G\) is a comparability graph with a
	given transitive orientation, a maximum clique of \(H_G\) can be computed by
	a longest-path dynamic program in time linear in the size of \(H_G\).
	Translating this clique back gives a
	maximum edge open packing of \(G\). The running-time analysis below gives the
	claimed bound.
\end{proof}

If the representation is not supplied, one can obtain it in polynomial time
using a standard permutation-graph recognition algorithm~\cite{Golumbic1980}.

\subsection{Running time}

Let \(n=|V(G)|\) and \(m=|E(G)|\). Assume that a permutation representation of
\(G\), equivalently the two linear orders \(<_{1}\) and \(<_{2}\), is given.
If \(m=0\), we return the empty packing immediately. Henceforth, assume
\(m\ge 1\).

From these two orders, we construct the product order \(\prec\) on \(V(G)\).
We also store the adjacency matrix of \(G\) and the comparison matrix of
\(\prec\). This takes \(O(n^2)\) time and allows every adjacency and
comparability test to be performed in constant time.

The graph \(A_G\) has one vertex for each ordered edge of \(G\). Hence,
\(|V(A_G)|=2m\). Since \(H_G=\overline{A_G}\), we also have
\(|V(H_G)|=2m\). Let \(N=2m\). We explicitly construct \(H_G\). There are
\(\binom{N}{2}=O(N^2)=O(m^2)\) pairs of ordered edges. For each pair
\((c,\ell)\) and \((d,r)\), we test compatibility in constant time. If
\(c=d\), we check whether \(\ell\neq r\) and \(\ell r\notin E(G)\). If
\(c\neq d\), we check whether \(c,\ell,d,r\) are distinct and whether
\(cd,cr,\ell d,\ell r\) are all nonedges of \(G\). Therefore, \(H_G\) can be
constructed in \(O(m^2)\) time.

During the same construction, we orient every edge of \(H_G\) according to the
orientation \(\triangleleft\). For same-centre compatible pairs, the
orientation is determined by comparing the two leaves in the product order.
For different-centre compatible pairs, the orientation is determined by
checking which endpoint pair is positioned completely before the other in
\(\prec\). Since all required comparisons are available in the comparison
matrix, each orientation decision takes constant time. Hence, the orientation
step also takes \(O(m^2)\) time.

By Theorem~\ref{thm:ag-permutation-cocomparability}, the graph \(H_G\) is a
comparability graph, and the construction above gives a transitive orientation
of \(H_G\). The vertices of every directed path form a clique, because
transitivity supplies an edge between every two nonconsecutive vertices of the
path. Conversely, the restriction of a transitive orientation to a clique is a
total order and therefore contains a directed path through all vertices of that
clique. Thus, a maximum clique has the same cardinality as a longest directed
path. Since a transitive orientation is acyclic, such a path can be found by
dynamic programming in a topological order in
\(O(|V(H_G)|+|E(H_G)|)\) time; see, for example,
Golumbic~\cite{Golumbic1980}. Since \(|V(H_G)|=2m\) and
\(|E(H_G)|=O(m^2)\), this step takes \(O(m^2)\) time.

Combining the above steps, the total running time is
\(O(n^2)+O(m^2)+O(m^2)+O(m^2)=O(n^2+m^2)\). Since
\(m\le \binom{n}{2}=O(n^2)\), this is bounded by \(O(n^4)\). Therefore, given
a permutation representation of \(G\), the
\textnormal{\textsc{Maximum Edge Open Packing}} problem can be solved on
permutation graphs in \(O(n^2+m^2)\le O(n^4)\) time.

\section{Interval Graphs}
\label{sec:interval}

In this section, we prove that the \textsc{Maximum Edge Open Packing} problem
is polynomial-time solvable on interval graphs. The proof uses the oriented
star-conflict graph \(A_G\) introduced in Section~\ref{sec:permutation}.
Recall that the vertices of \(A_G\) are the ordered edges of \(G\), and that
two vertices of \(A_G\) are adjacent precisely when the corresponding ordered
edges are incompatible. We write \(H_G=\overline{A_G}\). Thus, two vertices
of \(H_G\) are adjacent precisely when the corresponding ordered edges are
compatible.

By Lemma~\ref{lem:eop-independent-ag},
\(\rho_e^o(G)=\alpha(A_G)=\omega(H_G)\). We show that, when \(G\) is an
interval graph, \(H_G\) has a natural left-to-right transitive orientation.
Consequently, a maximum clique of \(H_G\), and hence a maximum edge open
packing of \(G\), can be obtained by a longest-path computation in this
orientation.

Let \(G\) be an interval graph, and fix an interval representation of \(G\).
For every \(v\in V(G)\), write \(I(v)=[L(v),R(v)]\). Two vertices are
adjacent in \(G\) precisely when their intervals intersect. Equivalently,
\(uv\notin E(G)\) if and only if \(R(u)<L(v)\) or \(R(v)<L(u)\). Since the
intervals are closed, intervals that meet at an endpoint are adjacent.
Accordingly, nonadjacency is expressed below by strict separation of
endpoints. We assume that the interval representation is supplied with the
input. If only the graph is supplied, a representation can be obtained as
part of standard linear-time interval-graph recognition
\cite{Golumbic1980}, without changing the asymptotic running time below.

\subsection{A left-to-right orientation of the compatibility graph}

Let \(e=(c,\ell)\) be a vertex of \(H_G\). Since \(c\ell\in E(G)\), the
intervals \(I(c)\) and \(I(\ell)\) intersect, and hence their union is an
interval. Define \(\lambda(e)=\min\{L(c),L(\ell)\}\) and
\(\rho(e)=\max\{R(c),R(\ell)\}\). We write
\(U(e)=I(c)\cup I(\ell)=[\lambda(e),\rho(e)]\).

We now orient every edge of \(H_G\). Let \(e=(c,\ell)\) and \(f=(d,r)\) be
adjacent vertices of \(H_G\), or equivalently, compatible ordered edges of
\(G\). First suppose that \(c=d\). Compatibility implies that \(\ell\neq r\)
and \(\ell r\notin E(G)\), so the leaf intervals are strictly separated. If
\(R(\ell)<L(r)\), orient \(e\to f\); otherwise,
\(R(r)<L(\ell)\), and orient \(f\to e\). Thus, compatible ordered edges with
the same centre are ordered according to the left-to-right order of their
leaves.

Now suppose that \(c\neq d\). Compatibility implies that
\(c,\ell,d,r\) are four distinct vertices and that no edge of \(G\) has one
endpoint in \(\{c,\ell\}\) and the other in \(\{d,r\}\). Therefore, every
interval in \(\{I(c),I(\ell)\}\) is disjoint from every interval in
\(\{I(d),I(r)\}\). Since each of \(U(e)\) and \(U(f)\) is an interval, the
two union intervals are strictly separated. If
\(\rho(e)<\lambda(f)\), orient \(e\to f\); otherwise,
\(\rho(f)<\lambda(e)\), and orient \(f\to e\). This assigns exactly one
direction to every edge of \(H_G\).

\begin{lemma}\label{lem:interval-orientation-transitive}
	The orientation of \(H_G\) defined above is transitive.
\end{lemma}

\begin{proof}
	Let \(e=(c,\ell)\), \(f=(d,r)\), and \(g=(a,b)\) be vertices of \(H_G\)
	such that \(e\to f\) and \(f\to g\). We prove that \(e\) and \(g\) are
	adjacent in \(H_G\), and that their edge is oriented as \(e\to g\).

	Suppose first that both arcs are of the different-centre type. Then
	\(\rho(e)<\lambda(f)\) and \(\rho(f)<\lambda(g)\). Hence,
	\(\rho(e)<\lambda(g)\), so \(U(e)\) lies strictly to the left of \(U(g)\).
	This strict separation implies both that the endpoint sets of \(e\) and
	\(g\) are disjoint and that no edge of \(G\) joins the two endpoint sets:
	if a vertex occurred in both, its interval would belong to both union
	intervals. Thus, \(e\) and \(g\) satisfy the different-centre compatibility
	condition, and the orientation rule gives \(e\to g\).

	Next suppose that \(e\to f\) is of the same-centre type and \(f\to g\) is
	of the different-centre type. Write \(e=(c,\ell)\) and \(f=(c,r)\). We
	have \(R(\ell)<L(r)\) and \(\rho(f)<\lambda(g)\). In particular,
	\(R(c)\leq \rho(f)<\lambda(g)\), while
	\(R(\ell)<L(r)\leq \rho(f)<\lambda(g)\). Therefore,
	\(\rho(e)<\lambda(g)\). As in the first case, strict separation guarantees
	that the two ordered edges have four distinct endpoints and that every cross
	pair is a nonedge of \(G\). Hence, \(e\) and \(g\) are compatible, and the
	orientation rule gives \(e\to g\).

	Now suppose that \(e\to f\) is of the different-centre type and
	\(f\to g\) is of the same-centre type. Write \(f=(d,r)\) and
	\(g=(d,b)\). We have \(\rho(e)<\lambda(f)\) and \(R(r)<L(b)\). It follows
	that \(\rho(e)<L(d)\) and \(\rho(e)<L(r)\leq R(r)<L(b)\). Consequently,
	\(\rho(e)<\lambda(g)\). Again, this strict separation ensures that the
	endpoint sets are disjoint and anticomplete. Thus, \(e\) and \(g\) are
	compatible, and the orientation rule gives \(e\to g\).

	Finally, suppose that both arcs are of the same-centre type. All three
	ordered edges then have the same centre, so we may write
	\(e=(c,\ell)\), \(f=(c,r)\), and \(g=(c,b)\). The orientation rule gives
	\(R(\ell)<L(r)\) and \(R(r)<L(b)\). Hence,
	\(R(\ell)<L(b)\), and therefore \(\ell b\notin E(G)\). The ordered edges
	\(e\) and \(g\) are compatible with the same centre, and their edge is
	oriented as \(e\to g\).

	In every case, \(e\to f\) and \(f\to g\) imply that \(eg\in E(H_G)\) and
	that \(eg\) is oriented as \(e\to g\). Hence, the orientation is
	transitive.
\end{proof}

\begin{corollary}\label{cor:interval-H-comparability}
	If \(G\) is an interval graph, then \(H_G=\overline{A_G}\) is a
	comparability graph.
\end{corollary}

\begin{proof}
	The orientation constructed above is a transitive orientation of \(H_G\) by
	Lemma~\ref{lem:interval-orientation-transitive}.
\end{proof}

\subsection{Algorithm and correctness}

If \(m=|E(G)|=0\), the algorithm returns the empty set, which is the unique
maximum edge open packing. Assume henceforth that \(m>0\). List the \(2m\)
ordered edges \((c,\ell)\) with \(c\ell\in E(G)\), construct \(H_G\) by
testing each pair for compatibility, and orient every edge of \(H_G\) by the
left-to-right rule above. We then compute a longest directed path in this
orientation.

For completeness, we explain why this last step yields a maximum clique. In
any transitive orientation, the vertices of a directed path are pairwise
adjacent, by repeated application of transitivity, and hence form a clique.
Conversely, the restriction of a transitive orientation to a clique is a
transitive tournament. Its vertices have a total left-to-right order and
therefore form a directed path in that order. Consequently, the maximum
clique size of \(H_G\) is exactly the maximum number of vertices on a directed
path.

Let \(K\) be the clique obtained from a longest directed path. Replace each
ordered edge \((c,\ell)\in K\) by its underlying edge \(c\ell\). A clique of
\(H_G\) cannot contain both \((u,v)\) and \((v,u)\), because those two ordered
edges are incompatible. Thus, this replacement preserves cardinality. Since
\(K\) is an independent set of \(A_G\), the proof of
Lemma~\ref{lem:eop-independent-ag} shows that the resulting set is an edge
open packing of \(G\). Conversely, every edge open packing admits at least one
orientation that gives a clique of \(H_G\). This correspondence need not be
one-to-one, because an isolated \(K_2\)-component of the induced star forest
may be oriented in either direction. Nevertheless, it preserves the optimum
value, and therefore the packing obtained from \(K\) is maximum.

\subsection{Running time}

Let \(n=|V(G)|\) and \(m=|E(G)|\). We store an adjacency matrix of \(G\),
which takes \(O(n^2)\) time and permits every adjacency test in
constant time. The graph \(H_G\) has \(2m\) vertices. Listing its vertices
takes \(O(m)\) time.

To construct \(H_G\), we examine all \(O(m^2)\) pairs of ordered edges. For
each pair, compatibility is tested in constant time using the adjacency
matrix. For ordered edges with the same centre, it is enough to test whether
their leaves are distinct and nonadjacent. For ordered edges with different
centres, we test whether their four endpoints are distinct and whether the
four cross pairs are nonedges of \(G\). Each compatible pair is oriented in
constant time by comparing interval endpoints. Hence, constructing and
orienting \(H_G\) takes \(O(m^2)\) time.

The transitive orientation is acyclic. After computing a topological ordering,
a standard dynamic program computes a longest directed path, together with
predecessor pointers for reconstructing it, in
\(O(|V(H_G)|+|E(H_G)|)=O(m^2)\) time. Therefore, the total
running time is \(O(n^2+m^2)\). Since
\(m\leq \binom{n}{2}\), the running time is at most \(O(n^4)\).

\begin{theorem}\label{thm:eop-interval-o4}
	Given an interval representation of an interval graph \(G\), the
	\textnormal{\textsc{Maximum Edge Open Packing}} problem can be solved in
	\(O(n^2+m^2)\) time. In particular, it can be solved in
	\(O(n^4)\) time in terms of \(n\).
\end{theorem}

\begin{proof}
	By Lemma~\ref{lem:eop-independent-ag},
	\(\rho_e^o(G)=\omega(H_G)\). By
	Corollary~\ref{cor:interval-H-comparability}, \(H_G\) has the transitive
	orientation constructed above. A maximum clique of \(H_G\) is therefore
	obtained from a longest directed path, and the preceding construction maps
	such a clique to a maximum edge open packing of \(G\) without changing its
	cardinality. The construction of \(H_G\), its orientation, the longest-path
	computation, and reconstruction of the packing take \(O(n^2+m^2)\) time.
\end{proof}

\section{Well-Partitioned Chordal Graphs}
\label{sec:wpcg}

In this section, we give a dynamic programming algorithm for the
\textsc{Maximum Edge Open Packing} problem on well-partitioned chordal
graphs. The algorithm works over a partition tree and uses the fact that an
edge open packing has at most two endpoints in any clique bag.

Let \(G\) be a connected well-partitioned chordal graph, and let \(T\) be a
partition tree of \(G\). Thus, the nodes of \(T\) are pairwise disjoint clique
bags whose union is \(V(G)\). If \(XY\in E(T)\), then the edges between \(X\)
and \(Y\) form the complete bipartite graph
\(\operatorname{bd}(X,Y)\times\operatorname{bd}(Y,X)\); nonadjacent bags of
\(T\) have no edges between them. We root \(T\) at an arbitrary bag \(R\).
For a bag \(X\neq R\), let \(p(X)\) be its parent. We write \(T_X\) for the
subtree rooted at \(X\), and \(G_X\) for the subgraph induced by the bags of
\(T_X\).

\subsection{Structural observations}

We begin with two consequences of Observation~\ref{obs:eop-stars}.

\begin{observation}\label{obs:wpcg-bag-two}
	If \(D\) is an edge open packing of \(G\), then
	\(|X\cap V(D)|\leq 2\) for every bag \(X\in V(T)\).
\end{observation}

\begin{proof}
	Every bag is a clique. If three vertices of \(X\) belonged to \(V(D)\),
	they would induce a triangle in \(G[V(D)]\), whereas
	Observation~\ref{obs:eop-stars} says that \(G[V(D)]\) is a disjoint union
	of induced stars. Hence, at most two vertices of a bag belong to \(V(D)\).
\end{proof}

\begin{observation}\label{obs:eop-induced-edge-selected}
	Let \(D\) be an edge open packing of \(G\). If \(u,v\in V(D)\) and
	\(uv\in E(G)\), then \(uv\in D\).
\end{observation}

\begin{proof}
	Suppose that \(uv\notin D\). Since \(u,v\in V(D)\), there are selected
	edges \(e_u,e_v\in D\) incident with \(u\) and \(v\), respectively. The
	two edges are distinct: if \(e_u=e_v\), then this common edge has endpoints
	\(u\) and \(v\), and hence equals \(uv\), contrary to \(uv\notin D\).
	Therefore, \(uv\) is a third edge joining an endpoint of \(e_u\) to an
	endpoint of \(e_v\), contradicting the definition of an edge open packing.
\end{proof}

It is convenient to orient every star component of \(G[V(D)]\) from its
centre towards its leaves. A \(K_2\)-component has two possible orientations,
and either one may be used. The dynamic program stores these two possibilities
separately. This convention removes any ambiguity when a \(K_2\)-component is
extended through an unprocessed bag.

\subsection{Boundary states and their semantics}

Fix a bag \(X\), and order its children as \(Y_1,\ldots,Y_k\). For
\(0\leq i\leq k\), let \(G_{X,i}\) be the subgraph induced by \(X\) together
with the bags in \(T_{Y_1},\ldots,T_{Y_i}\). In particular,
\(G_{X,0}=G[X]\) and \(G_{X,k}=G_X\).

A partial configuration on \(G_{X,i}\) is a pair \((D,Q)\) with the following
properties. First, \(D\subseteq E(G_{X,i})\) is an edge open packing of
\(G_{X,i}\). Second, \(Q\subseteq X\setminus V(D)\) contains at most one
pending vertex. A pending vertex is reserved to become an endpoint of a
selected edge in an unprocessed child subtree or on the parent side of \(X\).
Accordingly, a pending vertex has no neighbour in \(V(D)\). Finally, every
star component of \(G_{X,i}[V(D)]\) is oriented from its centre to its leaves.

The restriction to at most one pending vertex entails no loss of generality.
Indeed, if two vertices \(x,y\in X\) belong to the endpoint set of a completed
solution, then \(xy\in E(G)\), and
Observation~\ref{obs:eop-induced-edge-selected} forces \(xy\) to be selected.
Thus, two vertices of the same bag cannot both remain pending.

The state of a partial configuration records its interaction with \(X\). The
possible states are as follows.

\begin{enumerate}[label=\textnormal{(\roman*)}]
	\item The empty state \(\varnothing\), in which
	\(X\cap V(D)=Q=\varnothing\).

	\item For every \(x\in X\), the pending state \(P_x\), in which
	\(Q=\{x\}\) and \(X\cap V(D)=\varnothing\).

	\item For every \(x\in X\), the singleton states \(C_x\) and \(L_x\). In
	these states, \(Q=\varnothing\), \(X\cap V(D)=\{x\}\), and \(x\) is,
	respectively, the centre or a leaf of its oriented star component.

	\item For every unordered pair \(\{x,y\}\subseteq X\), the pair states
	\(C_{x\mid y}\) and \(C_{y\mid x}\). The first state means that
	\(X\cap V(D)=\{x,y\}\), the edge \(xy\) is selected, \(x\) is the centre,
	and \(y\) is a leaf. The second state has the symmetric meaning.
\end{enumerate}

For a state \(\sigma\), let \(B(\sigma)\) be the set of vertices of \(X\)
mentioned by the state. Thus, \(B(P_x)=B(C_x)=B(L_x)=\{x\}\),
\(B(C_{x\mid y})=B(C_{y\mid x})=\{x,y\}\), and
\(B(\varnothing)=\varnothing\). Notice that a vertex in \(B(\sigma)\) is an
endpoint of a selected edge unless \(\sigma=P_x\).

Let \(\DP_{X,i}[\sigma]\) be the maximum value of \(|D|\) over all partial
configurations on \(G_{X,i}\) having state \(\sigma\). If no such
configuration exists, its value is \(-\infty\). After all children have been
processed, we write \(\DP_X[\sigma]=\DP_{X,k}[\sigma]\).

\subsection{Initialization}

The table for \(G_{X,0}=G[X]\) is initialized as follows. We set
\(\DP_{X,0}[\varnothing]=0\) and \(\DP_{X,0}[P_x]=0\) for every \(x\in X\).
For each unordered pair \(\{x,y\}\subseteq X\), the edge \(xy\) is present
because \(X\) is a clique. We store both possible orientations by setting
\(\DP_{X,0}[C_{x\mid y}]=\DP_{X,0}[C_{y\mid x}]=1\). All singleton states
\(C_x,L_x\) are initialized to \(-\infty\), because no selected edge is
incident with a single boundary vertex when only the bag \(X\) has been
processed.

\subsection{Constant-size state representatives}

We next define a representative \(R(\sigma)\) that contains exactly the
information relevant to future merges. The vertices of \(B(\sigma)\) are
called boundary vertices. Some additional vertices, called private markers,
encode selected neighbours that have already been forgotten.

The representative of \(\varnothing\) is empty, while \(R(P_x)\) consists of
the isolated pending boundary vertex \(x\). For \(C_x\), introduce a private
marker \(x^{\ell}\), add the edge \(xx^{\ell}\), and prescribe \(x\) as its
centre and \(x^{\ell}\) as a leaf. For \(L_x\), introduce a private marker
\(x^{c}\), add the edge \(x^{c}x\), and prescribe \(x^{c}\) as the centre and
\(x\) as a leaf. Finally, \(R(C_{x\mid y})\) consists of the edge \(xy\), with
\(x\) prescribed as centre and \(y\) as leaf; \(R(C_{y\mid x})\) is defined
symmetrically.

The private markers are not vertices of \(G\), and their edges are not counted
by the dynamic program. They only record whether a boundary vertex is already
a centre or is already a leaf. A centre may receive further leaves, whereas a
leaf cannot receive another selected neighbour.

\subsection{Combining a child}

Suppose that \(Y=Y_i\) is the next child of \(X\). Let \(\sigma\) be a state
of the current table \(\DP_{X,i-1}\), and let \(\tau\) be a state of
\(\DP_Y\). Since the only edges between the two processed regions are the
edges between \(X\) and \(Y\), the cross edges forced by these states are
\(F_{\sigma,\tau}=E_G(B(\sigma),B(\tau))\). Equivalently,
\(F_{\sigma,\tau}\) is the complete bipartite graph between
\(B(\sigma)\cap\operatorname{bd}(X,Y)\) and
\(B(\tau)\cap\operatorname{bd}(Y,X)\). These edges are forced by
Observation~\ref{obs:eop-induced-edge-selected}: every vertex mentioned by a
state is intended to be an endpoint in the completed solution, including a
pending vertex once it is joined at this merge.

Form a constant-size graph \(R(\sigma,\tau)\) from the disjoint union of
\(R(\sigma)\) and \(R(\tau)\) by adding all edges of
\(F_{\sigma,\tau}\). The boundary vertices keep their actual names; all
private markers are distinct. We say that \((\sigma,\tau)\) is compatible if
the following conditions hold.

\begin{enumerate}[label=\textnormal{(\roman*)}]
	\item Every pending boundary vertex belonging to \(Y\) is incident with a
	new cross edge. Otherwise, it would be forgotten without ever becoming an
	endpoint of a selected edge.

	\item A pending boundary vertex belonging to \(X\) may remain isolated, in
	which case it stays pending. If it is incident with a cross edge, it ceases
	to be pending.

	\item After deleting a possible isolated pending vertex of \(X\), every
	connected component of \(R(\sigma,\tau)\) is a star.

	\item Each star component admits a choice of centre consistent with every
	role prescribed in \(R(\sigma)\) and \(R(\tau)\). A component with at least
	three vertices has its unique vertex of degree greater than one as centre.
	For a \(K_2\)-component, either endpoint may be its centre. A prescribed
	centre must be the chosen centre, and a prescribed leaf must be different
	from it.
\end{enumerate}

A merge is \emph{reservation-preserving}: its selected edge set is the union
of the two selected edge sets and \(F_{\sigma,\tau}\), while its pending set
consists of the possible pending vertex of \(X\) if that vertex remains
isolated. Every pending vertex in the child bag \(Y\), which is forgotten
after this step, must therefore become incident with a cross edge. Conditions
(i) and (ii) enforce precisely this convention.

This is an explicit constant-time test: the representative has at most four
boundary vertices and at most two private markers. In particular, it checks
all forced edges, rejects triangles and nonstar paths, and preserves both
orientations of every newly created \(K_2\)-component.

For every compatible pair, each allowed choice of centres determines a state
on \(X\). If the retained boundary is empty, the projected state is
\(\varnothing\). An isolated pending vertex \(x\in X\) gives \(P_x\). One
active vertex \(x\in X\) gives \(C_x\) or \(L_x\) according to its role. Two
active vertices \(x,y\in X\) give \(C_{x\mid y}\) or \(C_{y\mid x}\). Let
\(\Pi_X(\sigma,\tau)\) denote the set of all projected states obtained in this
way. Its size is bounded by a constant.

The new table is first initialized to \(-\infty\). For every
\(\sigma'\in\Pi_X(\sigma,\tau)\), we replace \(\DP_{X,i}[\sigma']\) by the
maximum of its current value and
\(\DP_{X,i-1}[\sigma]+\DP_Y[\tau]+|F_{\sigma,\tau}|\). The last term counts
exactly the newly selected edges between \(X\) and \(Y\).

The following interface lemma explains why the representative test contains
all information needed by the transition.

\begin{lemma}\label{lem:wpcg-interface}
	Let partial configurations realizing \(\sigma\) and \(\tau\) be given on
	\(G_{X,i-1}\) and \(G_Y\), respectively. Their reservation-preserving union,
	together with the forced edge set \(F_{\sigma,\tau}\), is a valid partial
	configuration if and only if \((\sigma,\tau)\) passes the representative
	compatibility test. In that case, its state on \(X\) belongs to
	\(\Pi_X(\sigma,\tau)\).
\end{lemma}

\begin{proof}
	The partition-tree definition implies that vertices strictly inside two
	different child subtrees are nonadjacent. Moreover, a vertex of \(X\) has
	neighbours in \(G_Y\) only in the bag \(Y\). Consequently, every new graph
	edge whose endpoints are active, or become active in this merge, has one
	endpoint in \(B(\sigma)\) and the other in \(B(\tau)\), and every such edge
	is included in \(F_{\sigma,\tau}\).

	Inside either processed region, all components are already oriented stars.
	A boundary centre may be incident with any number of forgotten leaves, and
	a boundary leaf already has exactly one selected neighbour. The private
	markers in singleton states, together with the prescribed roles in pair
	states, record precisely these possibilities. Thus, replacing the
	forgotten part of a boundary star by its marker changes neither the ways in
	which the boundary vertex may be extended nor the validity of any new cross
	edge.

	It follows that a new triangle, a nonstar path, two joined centres, or a
	leaf receiving a second selected neighbour occurs in the actual union if
	and only if it occurs in the representative. Conditions (i) and (ii) are
	exactly the requirements for pending vertices that are forgotten or
	retained. Conditions (iii) and (iv) are exactly the induced-star and role
	requirements. The final choice of roles on the retained vertices is, by
	definition, one of the states in \(\Pi_X(\sigma,\tau)\).
\end{proof}

\subsection{Correctness}

\begin{lemma}\label{lem:wpcg-dp-sound}
	Every finite entry \(\DP_{X,i}[\sigma]\) is realized by a partial
	configuration of the stated value and state.
\end{lemma}

\begin{proof}
	We use induction on the postorder of \(T\), and, for a fixed bag, on the
	number of children already processed. The assertion holds for
	\(\DP_{X,0}\) by the initialization: the empty and pending configurations
	contain no selected edge, while every pair state consists of the single
	selected edge inside \(X\) under one of its two orientations.

	Assume that the assertion holds for \(\DP_{X,i-1}\) and for the final table
	of \(Y_i\). Every transition accepted by the algorithm passes the
	representative test. Lemma~\ref{lem:wpcg-interface} therefore shows that
	the two realizing configurations, together with the forced cross edges,
	form a valid partial configuration with the projected state. The update
	adds exactly the number of new cross edges. Taking the maximum preserves a
	realizing configuration for every finite entry. This completes both
	inductions.
\end{proof}

\begin{lemma}\label{lem:wpcg-dp-complete}
	Let \(\widehat D\) be an edge open packing of \(G\), and orient each
	component of \(G[V(\widehat D)]\) as a star. For every bag \(X\) and every
	\(i\), define
	\(D_{X,i}=\widehat D\cap E(G_{X,i})\) and
	\(Q_{X,i}=(X\cap V(\widehat D))\setminus V(D_{X,i})\). Then
	\(\DP_{X,i}\) contains the state of \((D_{X,i},Q_{X,i})\) with value at
	least \(|D_{X,i}|\).
\end{lemma}

\begin{proof}
	Restricting an oriented star to a subset of its edges again gives an
	oriented star. Hence, \(D_{X,i}\)
	is an edge open packing of \(G_{X,i}\).

	If \(Q_{X,i}\) contained two vertices \(x,y\), then
	\(x,y\in X\cap V(\widehat D)\), and
	Observation~\ref{obs:eop-induced-edge-selected} would imply
	\(xy\in D_{X,i}\), a contradiction. The same argument shows that a pending
	vertex cannot coexist with an active endpoint in \(X\). Moreover, a pending
	vertex has no neighbour in \(V(D_{X,i})\), since every such edge would be
	forced into \(D_{X,i}\). Thus, \((D_{X,i},Q_{X,i})\) has one of the states
	defined above.

	We prove the table assertion simultaneously by induction on the postorder
	of \(T\) and, for each bag, on \(i\). When \(i=0\), the trace contains only
	edges with both endpoints in the clique \(X\). If the global endpoint set
	meets \(X\) in no vertex, the trace has state \(\varnothing\). If it meets
	\(X\) in one vertex, that vertex is pending and the trace has state \(P_x\).
	If it meets \(X\) in two vertices, their edge belongs to \(\widehat D\) by
	Observation~\ref{obs:eop-induced-edge-selected}, and the trace has one of the
	two pair states. Observation~\ref{obs:wpcg-bag-two} shows that these are the
	only cases, and they are exactly the finite entries created by the
	initialization.

	Now let \(i>0\), and put \(Y=Y_i\). By the induction on \(i\), the trace on
	\(G_{X,i-1}\) occurs in \(\DP_{X,i-1}\). By the postorder induction, the
	trace of the same global packing on \(G_Y\), including any vertex of \(Y\)
	whose selected incident edges lie outside \(G_Y\) as pending, occurs in
	\(\DP_Y\). The boundary sets of these two states are precisely
	\(X\cap V(\widehat D)\) and \(Y\cap V(\widehat D)\). Consequently, all edges
	of \(\widehat D\) between the two regions are exactly the forced cross edges
	\(F_{\sigma,\tau}\): every such cross edge is included, and every edge of
	\(F_{\sigma,\tau}\) belongs to \(\widehat D\) by
	Observation~\ref{obs:eop-induced-edge-selected}. The orientation inherited
	from \(\widehat D\) therefore witnesses that the representative test accepts
	these states and projects to the state of
	\((D_{X,i},Q_{X,i})\). The update gives that state value at least
	\(|D_{X,i}|\), completing both inductions.
\end{proof}

At the root \(R\), no pending vertex can be completed outside the processed
graph. Therefore, the accepting states are \(\varnothing\), every singleton
state \(C_x,L_x\), and every pair state \(C_{x\mid y},C_{y\mid x}\). Pending
states are rejected.

\begin{theorem}\label{thm:wpcg-dp-correct}
	The maximum value of an accepting state in the table of the root is
	\(\rho_e^o(G)\).
\end{theorem}

\begin{proof}
	By Lemma~\ref{lem:wpcg-dp-sound}, every accepting root entry is realized by
	an edge open packing of \(G\), so its value is at most \(\rho_e^o(G)\). Let
	\(D\) be a maximum edge open packing of \(G\). At the root, the pending set
	in Lemma~\ref{lem:wpcg-dp-complete} is empty. Hence, an accepting root state
	has value at least \(|D|=\rho_e^o(G)\). The two inequalities prove the claim.
\end{proof}

\subsection{Running time and reconstruction}

Let \(s_X=|X|\). The number of states on \(X\) is
\(1+3s_X+2\binom{s_X}{2}=O(s_X^2)\). For a tree edge \(XY\), the algorithm
examines \(O(s_X^2s_Y^2)\) pairs of states. Boundary-set membership is
preprocessed, so all forced cross edges and the representative graph are
constructed in constant time. Since the representative has constant size,
compatibility testing and projection also take constant time.

Consequently, the total time for all merges is
\(O(\sum_{XY\in E(T)}s_X^2s_Y^2)\). Since the bags partition \(V(G)\), we have
\(\sum_Xs_X=n\) and \(\sum_Xs_X^2\leq n^2\). The sum over the tree edges is at
most the corresponding sum over all unordered pairs of bags and is therefore
at most \((\sum_Xs_X^2)^2\leq n^4\). Initialization and preprocessing take
\(O(n^2)\) time and are dominated by the merge time.

For a disconnected graph, the algorithm is applied to the partition tree of
each connected component and the answers are added. If the component orders
are \(n_1,\ldots,n_t\), then \(\sum_j n_j^4\leq n^4\).

\begin{theorem}\label{thm:eop-wpcg-poly}
	Given a partition tree for every connected component of a well-partitioned
	chordal graph \(G\), its edge open packing number can be computed in
	\(O(n^4)\) time. A maximum edge open packing can
	be produced within the same time bound.
\end{theorem}

\begin{proof}
	Correctness follows from Theorem~\ref{thm:wpcg-dp-correct}. The preceding
	analysis gives the claimed time bounds, including reconstruction
	of an optimal edge set.
\end{proof}

\subsection{Example}

Let \(K=\{x,y\}\) be a clique, let \(S=\{a,b,c,d,e\}\) be an independent
set, and let \(E(G)=\{xy,xa,xb,xc,yd,ye\}\). This split graph has a partition
tree whose central bag is \(X_0=\{x,y\}\), with one singleton leaf bag for
each vertex of \(S\). We root the tree at \(X_0\).

For every singleton leaf bag, the local table contains the empty state and its
pending state, both with value \(0\). At \(X_0\), selecting \(xy\) creates the
two oriented pair states \(C_{x\mid y}\) and \(C_{y\mid x}\), each with value
\(1\). Starting with \(C_{x\mid y}\), merging the pending states of
\(a,b,c\) successively forces \(xa,xb,xc\). Each representative remains a
star centred at \(x\), so the state stays \(C_{x\mid y}\), and its value
becomes \(4\).

The pending state of \(d\) would force \(yd\). The representative would then
give the leaf \(y\) a second selected neighbour, so the merge is rejected. The
same holds for \(e\). Thus, the algorithm returns
\(D=\{xy,xa,xb,xc\}\), an edge open packing of size \(4\).

This solution is optimal. If both \(x\) and \(y\) are endpoints of selected
edges, then \(xy\) is selected, and one of \(x,y\) is the centre while the
other is a leaf. A star centred at \(x\) has at most the four edges
\(xy,xa,xb,xc\), while a star centred at \(y\) has at most the three edges
\(xy,yd,ye\). If only one of \(x,y\) is an endpoint, at most three edges can
be selected. Therefore, \(\rho_e^o(G)=4\). Notice that
\(x\) and \(y\) have disjoint neighbourhoods in \(S\), which is why all three
neighbours of \(x\) may accompany the clique leaf \(y\) in this example.

\section{Conclusion}
\label{sec:conclusion}

We studied the \textsc{Maximum Edge Open Packing} problem on permutation
graphs, interval graphs, and well-partitioned chordal graphs. For permutation
and interval graphs, the oriented star-conflict framework converts the problem
into a maximum-clique computation in a transitively oriented compatibility
graph. This yields \(O(n^2+m^2)\)-time algorithms when the corresponding
representations are given. In particular, the interval-graph result extends
the previously known tractability for proper interval graphs.

For well-partitioned chordal graphs, we developed a dynamic program over the
partition tree. Explicit boundary roles and pending states preserve both
orientations of a boundary \(K_2\)-component, while constant-size
representatives make every compatibility test local. The resulting algorithm
computes the optimum value in \(O(n^4)\) time; an
optimal packing can be reconstructed in the same time.
Since split graphs are well-partitioned chordal, this extends the known
tractability result for split graphs.

These results give further partial progress on polynomial-time solvability for
chordal graphs of unbounded clique number. Natural questions include improving
the \(O(n^4)\) worst-case bounds, identifying additional graph classes whose
oriented star-conflict graphs have exploitable structure, and extending the
partition-tree states to broader clique-based decompositions.

\bibliographystyle{plain}
\bibliography{EOP_Bib_per}

\end{document}